\documentclass[11pt]{article}
\usepackage{graphics,epsfig}
\usepackage[]{graphicx}
\usepackage{latexsym}

\begin{document}
\title{A Discussion on Particle Number and Quantum Indistinguishability}
\author{{\sc Graciela
Domenech}\thanks{%
Fellow of the Consejo Nacional de Investigaciones Cient\'{\i}ficas
y T\'ecnicas (CONICET)} \ {\sc and} \ {\sc Federico Holik}}

\maketitle

\begin{center}

\begin{small}
Instituto de Astronom\'{\i}a y F\'{\i}sica del Espacio (IAFE)\\
Casilla de Correo 67, Sucursal 28, 1428 Buenos Aires, Argentina\\

\end{small}
\end{center}

\vspace{1cm}

\begin{abstract}
\noindent

The concept of individuality in quantum mechanics shows radical
differences from the concept of individuality in classical
physics, as E. Schr\"{o}dinger pointed out in the early steps of
the theory. Regarding this fact, some authors suggested that
quantum mechanics does not possess its own language, and
therefore, quantum indistinguishability is not incorporated in the
theory from the beginning. Nevertheless, it is possible to
represent the idea of quantum indistinguishability with a first
order language using quasiset theory ($Q$). In this work, we show
that $Q$ cannot capture one of the most important features of
quantum non individuality, which is the fact that there are
quantum systems for which particle number is not well defined. An
axiomatic variant of $Q$, in which quasicardinal is not a
primitive concept (for a kind of quasisets called finite
quasisets), is also given. This result encourages the searching of
theories in which the quasicardinal, being a secondary concept,
stands undefined for some quasisets, besides showing explicitly
that in a set theory about collections of truly indistinguishable
entities, the quasicardinal needs not necessarily be a primitive
concept.
\end{abstract}
\bigskip
\noindent

\newtheorem{theo}{Theorem}[section]
\newtheorem{definition}[theo]{Definition}

\newtheorem{lem}[theo]{Lemma}

\newtheorem{prop}[theo]{Proposition}

\newtheorem{coro}[theo]{Corollary}

\newtheorem{exam}[theo]{Example}

\newtheorem{rema}[theo]{Remark}{\hspace*{4mm}}

\newtheorem{example}[theo]{Example}

\newcommand{\proof}{\noindent {\em Proof:\/}{\hspace*{4mm}}}

\newcommand{\qed}{\hfill$\Box$}

\newcommand{\ninv}{\mathord{\sim}} 

\begin{small}
\centerline{\em Key words:  quasisets, particle number,
quasicardinality, quantum indistinguishability.}
\end{small}

\bibliography{pom}

\newtheorem{axiom}[theo]{Axiom}

\newpage

\section{Introduction}

It is a well established result that the concept of individuality
in quantum mechanics clashes radically with its classical
counterpart. While in classical physics particles can be
considered as individuals without giving rise to consistence
problems, in quantum mechanics this is not the case.
Contradictions arise if one intends to individuate elementary
particles. The responses to this problem range from the claim that
there are no elementary particles at all to the assertion that
there are particles but they are intrinsically indistinguishable
(i.e., indistinguishable in an ontological sense).

The problem is solved in the formalism by imposing symmetrization
postulates (in non relativistic quantum mechanics) or,
equivalently, by imposing commutation relations on the
creation-annihilation operators (in quantum field theory). But
these solutions have a flow, well illustrated in \cite{Why
quasisets} and \cite{fyk}. The objection is that all these
approaches make use of a mathematical trick referred to in
\cite{Why quasisets} as the Weyl's strategy. This trick consists
in treating particles as if they were individuals and then
imposing symmetrization assumptions, thus masking individuality.
We find this a source of conceptual confusion and mathematical (or
axiomatic) redundancy, for if particles are not individuals (and
this implies that they cannot be labeled in the usual way), one
could ask why do permutations make sense. Is the symmetrization
postulate really necessary or is it a necessity of our own (and
inadequate) language?

Many authors pointed out the importance of developing alternative
ways to describe quantum indistinguishability, reproducing the
results obtained by standard techniques, but assuming in every
step of the deduction that elementary particles of the same class
are intrinsically indistinguishable from the beginning (see, for
example, \cite{Heinz Post}, \cite{Why quasisets} and \cite{fyk}),
without making appeal to Weyl's strategy variants. Another claim
is that quantum mechanics does not possess its own language, but
it uses a portion of functional analysis which is itself based on
set theory, and thus finally related to classical experience. This
statement was posed by Y. Manin \cite{Manin2}, the Russian
mathematician who suggested that standard set theories (as
Zermelo-Fraenkel, $ZF$) are influenced by every day experience,
and so it would be interesting to search for set theories which
inspire its concepts in the quantum domain. This is known as the
Manin's problem \cite{Manin}. In this spirit, and looking for a
solution to the Manin's problem a quasiset theory ($Q$ in the
following) was developed \cite{Why quasisets}, \cite{Un estudio}.
We base our work in the axiomatic system as presented in \cite{Un
estudio}.

Quasiset theory seems to be adequate to represent as ``sets'' of
some kind (quasisets) the collections of truly indistinguishable
entities. This aim is reached in $Q$ because equality is not a
primitive concept, and there exists certain kinds of urelemente
(m-atoms) for which only an indistinguishability relationship
applies. So, in $Q$, non individuality is incorporated by
proposing the existence of entities for which it has no sense to
assert that they are identical to themselves or different from
others of the same class.

$Q$ contains a copy of Zermelo-Fraenkel set theory plus Urelemente
($ZFU$). These Urelemente are called M-atoms. This feature divides
the theory in two parts. One region involves only the elements of
$ZFU$, and the other one contains quasisets whose elements can be
truly indistinguishable entities. Quasisets containing only
indistinguishable elements are called ``pure quasisets''. We will
refer to the $ZFU$ copy of quasisets as ``the classical part of
the theory '', as in \cite{Un estudio}. Indistinguishability is
modeled in this theory using a primitive binary relation $\equiv$
(indistinguishability) and a new class of atoms, called m-atoms,
which stands for expressing the existence of quanta in the theory
(\cite{Un estudio}). So, in the frame of $Q$, when we speak of
m-atoms of the same class, the only thing that we can assert about
them is that they are indistinguishable, and nothing else makes
sense, for expressions like $x=y$ are not well formed formulas for
them. This is to say that we cannot make assertions about their
identity, i.e., it has no sense to say that an m-atom is equal or
different of other m-atom of the same class.

It is important to remark that in $Q$, indistinguishability does
not imply identity, and so it is possible that even being
indistinguishable, two m-atoms belong to different quasisets, thus
avoiding the collapse of indistinguishability in classical
identity \cite{Un estudio}.

$Q$ is constructed in such a way that allows the existence of
collections of truly indistinguishable objects, and thereof it is
impossible to label the elements of pure quasisets. For this
reason, the construction used to assign cardinals to sets of
standard $ZFU$ theories cannot be applied any more.

But even if electrons are indistinguishable (in an ontological
sense), every physician knows that it is possible to assert that
for example, a Litium atom has three electrons. It is for that
reason that $Q$ should allow quasisets to have some kind of
associated cardinal. In $Q$ this is solved postulating that a
cardinal number is assigned to every quasiset (remember that there
is a copy of $ZFU$ in $Q$). Some other properties of the standard
cardinal are postulated too. This rule for the assignment of
cardinals uses a unary symbol $qc()$ as a primitive concept. Here
we will call ``axioms of quasicardinality '' the collection of
axioms related to the $qc()$ functional letter. So in $Q$, the
quasicardinal is a primitive concept alike $ZF$, in which the
property that to every set corresponds a single cardinal number
can be derived from the axioms \cite{Halmos}.

On the other hand, this form of introducing the quasicardinal
implies that every quasiset has an associated cardinal, i.e.,
every quasiset has a well defined number of elements. But the idea
that an aggregate of entities must necessarily have an associated
number which represents the number of entities is based in our
every day experience. As we shall discuss below, (section
\ref{s:QaR}) there are quantum systems for which it is not allowed
to assign a number of particles in a consistent manner. These
systems can be found in states which are not eigenstates of the
particle number operator. Another problematic situation (which
will be not discussed in this work) is that of the frame
dependence of the particle number operator. In relativistic
quantum field theory, the vacuum state in a Minkowskian frame
(which has particle number equal to zero in this frame) is seen as
a ``plenty of particles'' state for a Rindler observer
\cite{Rindler}, thus cardinality seeming frame dependent.

In \ref{s:QaR} we will show that from the fact that in $Q$ every
quasiset has an associated cardinal (its quasicardinal), it
follows that it is not possible to represent systems which are not
eigenstates of the particle number operator as quasisets. In
\ref{s:Individuality} we discuss the relationship of the concept
of particle number with the measurement process. Based on the
discussions previously posed, in \ref{s:Reformulation} we explore
the possibility of modifying $Q$ in such a way that the
quasicardinal becomes a derived concept. Finally, in
\ref{s:Conclusions} we draw our conclusions.

\section{Quantity in Quantum Mechanics}\label{s:QaR}

As is well known, performing a single measurement in a quantum
system does not allow to attribute the result of this measurement
to a property which the system possesses before the measurement is
performed without giving rise to serious problems
\cite{Mittelstaedt}. What is the relationship between this fact
and the quantity of particles in a quantum system? Take for
example an electromagnetic field (with a single frequency for
simplicity) in the following state:
\begin{equation}\label{e:GATOS}
\mid\psi\rangle=\alpha\mid 1\rangle+\beta\mid 2\rangle
\end{equation}
where $\mid 1\rangle$ and $\mid 2\rangle$ are eigenvectors of the
particle number operator with eigenvalues $1$ and $2$
respectively, and $\alpha$ and $\beta$ are complex numbers which
satisfy $\mid\alpha\mid^{2}+\mid\beta\mid^{2}=1$. If a measurement
of the number of particles of the system is made, one or two
particles will be detected, with probabilities
$\mid\alpha\mid^{2}$ and $\mid\beta\mid^{2}$ respectively. And any
other possibility is excluded. Suppose that in a single
measurement two particles are detected. What allows us to conclude
that the system had two particles before the measurement was
performed? The assertion that the number of particles is varying
in time because particles are being constantly created and
destroyed is also problematic, because it assumes that at each
instant the number of particles is well defined. Only in case that
it is known with certainty that the system is in an eigenstate of
the particle number operator we can say that the system has a well
defined cardinal. There would be no problem too if it is known
with certainty that the system is prepared in an statistical
mixture. In this case, the corresponding density operator would
be:
\begin{equation}
\rho_{m}=\mid\alpha\mid^{2}\mid 1\rangle\langle
1\mid+\mid\beta\mid^{2}\mid 2\rangle\langle 2\mid
\end{equation}
where the subindex ``m'' stands for statistical mixture. But the
density operator corresponding to (\ref{e:GATOS}) is:
\begin{equation}
\rho=(\alpha\mid1\rangle
+\beta\mid2\rangle)(\alpha^{*}\langle1\mid+\beta^{*}\langle2\mid)
\end{equation}
which is the same as:
\begin{equation}\label{e:MGATO}
\rho=\mid\alpha\mid^{2}\mid 1 \rangle\langle
1\mid+\mid\beta\mid^{2}\mid 2\rangle\langle
2\mid+\alpha\beta^{*}\mid 1\rangle\langle2\mid+\alpha^{*}\beta\mid
2\rangle\langle 1\mid
\end{equation}
The presence of interference terms in the last equation implies
that difficulties will appear in stating that, after a single
measurement, the system has the quantity of particles obtained as
the result of the measurement. In this case, the incapability of
knowing the particle number would not come from our ignorance
about the system, but from the fact that in this state, the
particle number is not even well defined.

Taking into account these considerations, it is worth asking: is
it possible to represent a system prepared in the state
(\ref{e:MGATO}) in the frame of quasiset theory? Which place would
correspond to a system like (\ref{e:MGATO}) in that theory? If
such system could be represented as a quasiset, then it should
have an associated quasicardinal, for every quasiset has it. But
this does not seem to be proper, considering what we have
discussed in this section. It follows that it does not appear
reasonable to assign a quasicardinal to every quasiset if quasiset
theory has to include all bosonic and fermionic systems (in all
their possible many particle states). Therefore, a system in the
state (\ref{e:MGATO}) cannot be included in $Q$ as a quasiset.
Yet, it would be interesting to study the possibility of including
systems in those states (such as (\ref{e:MGATO})) in the
formalism. A possible way out is to reformulate $Q$ in such a way
that the quasicardinal is not to be taken as a primitive concept,
but as a derived one, turning into a property that some quasisets
have and some others do not (in analogy with the property ``being
a prime number'' of the integers). Those quasisets for which the
property of having a quasicardinal is not satisfied, would be
suitable to represent quantum systems with particle number not
defined (such as (\ref{e:MGATO})). This property would also fit
well with the position that asserts that particle interpretation
is not adequate in, for example, quantum electrodynamics. With
such a modification of $Q$, a field (in any state) could always be
represented as a quasiset, avoiding the necessity of regarding the
field as a collection of classical ``things''. On the contrary,
the field would be described by a quasiset which has a defined
quasicardinal only in special cases, but not in general. And for
that reason this quasiset could not be interpreted as simply as a
collection of particles (because it seems reasonable to assume
that a collection of particles, indistinguishable or not, must
always have a well defined particle number).

Another possibility would be the introduction of a vector space
similar to a Hilbert space, but constructed using the non
classical part of $Q$. This option will be considered in a
forthcoming paper. In the following section we discuss the concept
of particle number as the result of a process (the measurement
process). We will discuss its relationship with the idea of
individuality suggested by experiments, and relate it with the
possibility of developing quasicardinal as a derived concept.

\section{Particle number as the result of a process}\label{s:Individuality}

In the last section, we suggested that the development of a
$Q$-like theory in which quasicardinal is a derived concept could
be useful if one aims to represent as quasisets, quantum systems
with particle number not defined. In this section, we discuss the
experimental relation of the concept of ``particle number'' and
find new arguments for the development of quasicardinal as a
derived concept. We start posing the question: In which sense do
we talk about quantum systems composed, for example, of a single
photon? We certainly know about the existence of the
electromagnetic field, and that this field obeys the rules of
quantum mechanics. How do we decide if the field is in a single
photon state or not? What do we mean when we use the words
``single photon''? These questions find an answer in our
laboratory experience, i.e., making measurements on the system.
The measurement process (which in the case of photons could be
described by the theory of interaction of the electromagnetic
field with matter) allows us to construct an idea of individuality
which in time allows us to speak about the photon as a particle.
In a similar way, and always mediated by a measuring process, we
talk about the other particles (electrons, protons, etc.). But
these corpuscular features of quantum systems differ notably from
the classical ones, and though experiments suggest an idea of
individuality, it is well established that this does not enable us
to consider particles as individuals, at least not in an equal
sense to classical individuality. Elementary particles cannot be
considered as individuals, as E. Schr\"{o}dinger pointed out in
the early days of quantum mechanics \cite{What is}. In spite of
these difficulties, we continue speaking about photons, electrons,
etc., using a jargon which has a lot of points in common with
classical physics, source of conceptual confusion.

Let us consider an example to illustrate how particle number
arises as a result of the measuring process. A photoelectric
detector consists in its fundamental aspects of an atom that can
be ionized due to the interaction with the electromagnetic field.
The signal (a current originated by the ionized atom) must be
amplified in order to be detected. The amplified signal is a
(\emph{macroscopical}) current, and we say that the intensity of
this current is proportional to the quantity of ``absorbed
photons'' in the volume of the detector (in practice, composed of
many atoms). In the limit of single photon states, we would
observe a single current pulse each time a photon is detected.
Thus, we see that once the detection mechanisms are considered, it
is possible to assign to some quantum systems an associated
number, which represents the ``particle number''. It is important
to point out that the so called ``particle number'' only appears,
in general, after the measurement process is performed. In other
words, experimental experience is the condition for particle
interpretation of the corresponding theoretical concepts (for
example, of the particle number operator). And we have already
mentioned that the measurement process almost always implies the
modification of the original state, and that the result of the
measurement cannot be attributed in general to a property
pertaining to the system before the measure is performed. In
particular, it is not true that a particle number can be always
assigned in a consistent manner, as we saw in the last section.
Thus, counting the quantity of elements in quantum mechanics (here
understood as measuring particle number) is qualitatively
different from counting the quantity of elements of a classical
system. In particular, in quantum mechanics the system is usually
destroyed or modified when counted, alike the classical case,
where the counting process can be made in principle without
disturbing the system.

Nevertheless, we know that there exist in nature systems for which
it is possible to assign a cardinal in a consistent manner (a well
defined number of particles, as for example the electrons of a
Litium atom, or single photon states). But they cannot simply be
considered as aggregates of individuals as if they were
distinguishable. This is to say that there should exist the
possibility of counting without distinguishing. If this were not
the case, physicians would have never talked about something like
``number of indistinguishable particles''. Here, the word
``counting'' is taken in the sense of assigning in a consistent
manner a ``number of elements'' to a system which is not so simple
as an individuals aggregate. For example, we could count how many
electrons has an Helium atom imagining the following process
(perhaps not the best, but possible in principle). Put the atom in
a cloud chamber and use radiation to ionize it. Then we would
observe the tracks of both, an ion and an electron. It is obvious
that the electron track represents a system of particle number
equal to one and, of course, we cannot ask about the identity of
the electron (for it has no identity at all), but the counting
process does not depend on this query. The only thing that cares
is that we are sure that the track is due to a single electron
state, and for that purpose, the identity of the electron does not
matter. If we ionize the atom again, we will see the track of a
new ion (of charge $2e$), and a new electron track. Which electron
is responsible of the second electron track? This query is ill
defined, but we still do not care. Now, the counting process has
finished, for we cannot extract more electrons. The process
finished in two steps, and so we say that an Helium atom has two
electrons, and we know that, as the wave function of the electrons
is an eigenstate of the particle number operator, no problem of
consistence will arise in any other experiment if we make this
assertion. In \cite{Dalla Chiara} Dalla Chiara and Toraldo di
Francia had already noted that we know experimentally that the
Helium atom has two electrons, because we can ionize it and
extract two separate electrons. They were looking for experimental
and theoretical grounds for developing a Quaset theory (for a
comparison between $Q$ and Quaset theory see \cite{una compara}).
In the same way, we explore on experimental experience to justify
the search of a $Q$-like theory in which quasicardinal is a
derived concept.

From the example of the ionized Helium atom, we find that the
process of counting the elements of a given ``collection''
extracting them one by one can be applied to some quantum systems
without giving rise to serious contradictions. Then, we should be
able of counting the Urelemente of some quasisets too. A radical
difference between counting the electrons of an atom and counting
the elements in a collection of classical objects in the way shown
above is that, in the classical case we can ask about the identity
of the extracted element at each step while, in the case of the
atom, this cannot be done. But this fact, does not alter the
essence of the counting process and we will exploit this fact. In
the following section, we will translate this idea to the language
of $Q$. As we have already mentioned, $Q$ describes collections of
truly indistinguishable objects as quasisets and the quasicardinal
is introduced as a primitive concept. The latter is justified
arguing that indistinguishability prohibits well ordering, and for
that reason the possibility of counting $\grave{a}\ la$ $ZF$. We
agree that quasisets cannot be counted in the same form as in
$ZF$, but our point is that it should be interesting to search for
other ways of counting, motivated by physical examples.

Experiments on quantum systems sometimes show corpuscular
features, and this suggests an idea of individuality. This idea is
a base for developing the concept of particle and
\emph{subsequently}, the notion of particles aggregate. In analogy
with this, in the next section we will develop a notion of
``individual quasiset'', which will be used as a base for
developing the notion of quantity of elements of a certain kind of
given quasisets. This is to be done without making appeal to
classical individuality and avoiding the introduction of
quasicardinal as a primitive concept (and avoiding
quasicardinality axioms). As in the laboratory, where the system
to be counted is submitted to a process which transforms the
original state of the system, and only from this modification is
that we talk about particle number, the axiomatic variant exposed
in the rest of this article submits some quasisets (called finite
quasisets) to a process which assigns them a quasicardinal in a
consistent manner. We want to consider quasicardinal as the result
of a process.

\section{Reformulation of Quasicardinality Axioms}\label{s:Reformulation}

In section \ref{s:QaR} we suggested that it would be interesting
to enrich $Q$ in such a way that systems with not defined particle
number, could be described with the formalism. We suggested that a
possible way to follow is to modify $Q$. To do this, we search for
a Q-like theory in which the quasicardinal is not to be taken as a
primitive concept, alike $Q$. In such a theory, to have a well
defined quasicardinal could be a property that some quasisets
possess and others do not, thus allowing the existence of
quasisets with their quantity of elements not well defined. In
section \ref{s:Individuality} we discussed the links between the
concept of particle number and the measurement process. There we
recalled that it is possible to assign a particle number in a
consistent manner to systems which are not as simply as individual
aggregates (for example, to electrons in an atom). We said that
this fact suggested that a theory in which quasicardinal is a
derived concept could be conceived. In the rest of this section we
explore this possibility.

The idea behind the formal construction shown below is very
simple. Once we have developed the intuitive idea of a quasiset
with one element (the singleton), this is used to count the
urelemente of a given quasiset extracting them, one by one,
without knowing which one is extracted at each time. The challenge
is to express this simple idea in a first order language without
identity, and using only the axioms of $Q$ which do not contain
the primitive unary functional letter $qc()$. The formal
construction shown below enriches the questions posed in
\cite{Sant'Anna} about the physical and philosophical implications
of the ``labeling process'' described there, which alike our
construction, makes use of the quasicardinality axioms. From now
on, we use the same notation as in \cite{Un estudio}, and also
results and definitions therein.

\subsection{Singletons}\label{s:S}

In $ZF$ set theory it is possible to construct, for a particular
set, another set whose only element is the given set. This is to
say: if $A$ is a set, there exists $\{A\}$, the set whose only
element is $A$. Of course, the cardinal of $\{A\}$ is $1$. Its
analogous in $Q$ is the quasiset:$$[A]=_E[z:z\equiv A]$$ the
quasiset of all the $z$ indistinguishable from $A$. $=_E$ stands
for the extensional equality defined in \cite{Un estudio}. There,
it is defined as a binary relation and it is a derived concept
(alike $\equiv$, which is a primitive one). It collapses into
classical equality when applied to elements of $ZFU$ or M-atoms,
and it is not defined for m-atoms (which are all
indistinguishable).

In the frame of $Q$ it is not necessarily true that $[A]$ has a
quasicardinal equal to $1$, for it could contain other elements
indistinguishable from $A$, (apart from $A$ itself). Our aim in
this section is to construct a singleton analogous without using
the quasicardinality axioms of quasiset theory (i.e., the axioms
listed in \cite{Un estudio} which make use of $qc()$ functional
letter). To do this, we recall the following construction
(\cite{Un estudio}): suppose that $X$ is such that
$X\neq_E\emptyset$ and $\exists x(x\in X)$. If $\wp (X)$ stands
for the collection of all subquasisets of $X$, then there exists
the quasiset of the $a\subseteq X$ such that $x\in a$:
$$A_x=_E[a\in \wp (X):x\in a]$$ and therefore $$<x>=_E\cap_{a\in
A_{x}}a$$ exists too. So $<x>$ is the intersection of $A_x$.
Taking this fact into account, it is possible to give the
following definitions:
\begin{definition}
Given $X\neq_E\emptyset$ and $x\in X$ $$A_x=_E[a\in \wp (X):x\in
a]$$
\begin{quote} \emph{}\end{quote}
\end{definition}
\begin{definition}
Given $X\neq_E\emptyset$ and $x\in X$ $$<x>=_E\cap_{a\in A_{x}}a$$
 is the singleton of x relative to X.
\end{definition}
In \cite{Un estudio}, $<x>$ is a strong singleton.  The
interesting point is that we can show that the singleton defined
above satisfies the following properties:
\begin{prop}
In case that $X$ is a set (i. e. an element of the copy of $ZFU$),
if $x\in X$ it follows that $$<x>=_E\{x\}$$
\end{prop}
\begin{proof}
Straightforward from the axioms of Q. \qed \\
\end{proof}

In the following proposition, ``$Q(X)$'' stands for ``$X$ is a
quasiset''.

\begin{prop}\label{e:LABASE}
Let $X$ be such that $Q(X)$, $x\in X$, $\alpha$ such that
$Q(\alpha)$ and $\alpha\subseteq <x>$. Then, $$\alpha=_E\emptyset
\vee \alpha=_E <x>$$
\end{prop}
\begin{proof}
There are two possibilities: $x\in\alpha$ or $x\notin\alpha$.
Suppose the first holds. From $\alpha\subseteq<x>\subseteq X$, it
follows that $\alpha\in \wp (X)$ and if $x\in\alpha$ then
$\alpha\in A_{x}$. But in this case, by construction of $<x>$, it
follows that $<x>\subseteq\alpha$ and therefore $<x>=_E\alpha$. On
the other hand, if we suppose that $x\notin\alpha$, then:
$$x\in(<x>\setminus\alpha)\subseteq<x>\subseteq X
\longrightarrow(<x>\setminus\alpha)\in
A_{x}\longrightarrow<x>\subseteq(<x>\setminus\alpha)$$ So we have
$<x>=_E(<x>\setminus\alpha)$. But if this were true, suppose that
$z$ is such that $z\in\alpha$, then $z\in<x>$ and therefore $z$
would belong to $(<x>\setminus\alpha)$ and thus
$\neg(z\in\alpha)$, which is a contradiction. For this reason it
follows that $\forall z(z\notin\alpha)$ but this is to say that
$\alpha=_E\emptyset$. \qed
\\
\end{proof}
\\The last property states that $<x>$
only admits $\emptyset$ and $<x>$ itself as subquasisets. This is
a suitable property for a quasiset of quasicardinal $1$. But the
last assertion has to be taken intuitively for, at this step, we
have not given a rigorous definition of quasicardinal.
Notwithstanding, it is reasonable to interpret singletons as
quasisets with only one element, because the last two properties
suggest that they are the natural extension of the $ZF$ singleton.
In the following sections we will use these intuitive ideas to
define the quasicardinal, without making use of the axioms related
to the $qc()$ functional letter.

\subsection{Descendant Chains}

Below, we will make use of the following proposition: (were
$X\setminus<y>$ stands for the difference between the quasiset $X$
and the singleton $<y>$)

\begin{prop}\label{e:ddx}
$$\forall X(Q(X)\wedge X\neq_E\emptyset\longrightarrow\exists
y(y\in X\wedge X\setminus<y>\subset X))$$
\end{prop}

\begin{proof}
Let us see first that $X\neq_E\emptyset\longrightarrow\exists y
\in X$. If this were not true, then $\forall y$ we would have
$\neg(y\in X)$. But in this case it is easy to convince oneself
that $\forall z(z\in X\longleftrightarrow z\in\emptyset)$ and so
$X=_E \emptyset$.\\ Using that $\exists y( y\in X)$ it is possible
to construct $<y>$. On the other hand $X\setminus<y>\subseteq X$
by definition of inclusion. But in this case, $y\in X \wedge
\neg(y\in X\setminus<y>)$ , so we have $X\setminus<y>\subset X$.
\qed
\\
\end{proof}

Taking into account the last proposition, given a non empty
quasiset $X$, we can imagine the following process. Let us define
first the notion of ``direct descendent of $X$''. If
$X\neq\emptyset$ then $\exists z\exists_Q Y (z \in X \wedge Y=_E
(X\setminus<y>)\wedge Y\subset X)$ and a quasiset $Y$ with this
property will be called a direct descendent of $X$. In the last
phrase, ``$\exists_Q Y$'' stands for ``there exists a quasiset
$Y$''.

\begin{definition}

 $$DD_X(Y)\longleftrightarrow\exists z ( z \in X \wedge Y=_E X
 \setminus <z> )$$

$DD_X(Y)$ is read as: Y is a direct descendent of X.
\end{definition}

So, given $X\neq\emptyset$ (by proposition \ref{e:ddx}) there
always exists a direct descendent of $X$, let us call it $X^-$,
which satisfies:
             $$X \supset (X^-)$$
It could happen that $X^-=\emptyset$ or not. If
$X^-\neq_E\emptyset$ it follows that there exists a direct
descendent of $X^-$, call it $X^{--}$. Then we have: $$X \supset
X^{-} \supset X^{--} $$

Going ahead with this process, it could be the case that this
chain of inclusions stopped (in case the last quasiset so obtained
be the empty quasiset), or that it has no end. So we could
conceive two qualitatively different situations:

Situation $1$: $$X \supset X^{-} \supset X^{--} \supset X^{---}
\supset \cdots \cdots$$ (the inclusions chain continues
indefinitely)

Situation $2$: $$X \supset X^{-} \supset X^{--} \supset X^{---}
\supset \cdots \cdots \supset \emptyset$$ (the inclusion chain
ends in the empty quasiset).

It is not clear for us how to assure the existence of these chains
without making appeal to the quasicardinality axioms. But in the
following we will assume that given any arbitrary quasiset, we can
always follow this procedure taking into account the example of
the Helium atom given in section \ref{s:Individuality}. Suppose we
want to describe the collection of the electrons in the atom as a
pure quasiset $X$. Then, when the atom is ionized, the free
electron track and the electrons in the ion could be interpreted
as a singleton $<x>$ extracted to $X$, and a quasiset $ X^{-}=_E
X\setminus<x>$ respectively. If the atom is ionized again, we
would obtain another singleton $<y>$ extracted to $X^{-}$ and
$X^{--}=_E \emptyset$, with analogous interpretations. Then we
obtain the chain of inclusions $X\supset X^{-} \supset \emptyset$
(as in situation $2$ described above) which reaches the empty
quasiset in two steps, expressing the fact that an Helium atom has
two electrons. In the following, we will postulate the existence
of these chains. To express these ideas in formal terms, it is
necessary to translate them to first order language using the
axiomatic of $Q$. It is important to remark that in this work, we
do not make appeal to the axioms of quasicardinality alike
Sant'Anna in \cite{Sant'Anna}. With this aim we make the following
definition:

\begin{definition}\label{e:cdx}
 $$CD_X(\gamma)\longleftrightarrow$$ $$( \gamma \in \wp (\wp
 (X)) \wedge X \in\gamma  \wedge \forall z \forall y (z\in \gamma
 \wedge y\in \gamma\wedge z\neq_{E} y\longrightarrow(z\supset y
 \vee y\supset z))$$
 $$\wedge\forall z(z\in \gamma\wedge
 z\neq_{E}\emptyset\longrightarrow\exists y(y\in\gamma\wedge
 DD_z(y)\wedge\forall w(w\in\gamma\wedge DD_z(w)\longrightarrow
 w=_E y))))$$\
$CD_X(\gamma)$ is read as: $\gamma$ is a descendant chain of $X$.
\end{definition}
Definition \ref{e:ddx} says that $\gamma$ is a descendant chain of
$X$ if and only if it is a collection of enclosed subquasisets of
$X$ (each one is included in its predecessor) such that if
$X\in\gamma$, $z\in\gamma$ and $z\neq\emptyset$, then there exists
a direct descendant of $z$ which belongs to $\gamma$. The
definition of descendant chain was made to express in the language
of $Q$ (minus quasicardinality axioms) situations $1$ and $2$
considered above. At the end of this section we will discuss that
this concept could also be used to describe another situations.

Using this definition we will introduce the following postulate,
which asserts that for any given quasiset there exists at least
one descendant chain.

\begin{axiom}\label{a:H1}
(Axiom of descendent chains) $$\forall_Q X
(X\neq_E\emptyset\longrightarrow\exists \gamma (CD_X(\gamma)))$$
\end{axiom}

This postulate will be used below to construct the quasicardinal
as a derived concept.

\subsection{Finite Quasisets}\label{s:QFIN}

In this section, we will introduce the notion of finite quasisets.
It is based on the notion of descendant chain given above and
involves quasifunctions \cite{Un estudio}. We will use the
notation $qf(F)$ to express that ``$F$ is a quasifunction''. A
quasifunction is a collection of ordered pairs, and if the pair
$<x;y>$ belongs to $F$, we will write $F(x)=_E y$ whenever $y$ is
itself a quasiset. Note that this notation has no sense if $y$ is
an $m-atom$, because extensional equality is not defined for such
Urelemente. In the following definition, $\omega$ is the quasiset
which contains $\emptyset$ and all its successors \cite{Halmos}.
If $n^{+}$ is a successor of $n$, then $n^{+}=_E n\cup \{n\}$.

\begin{definition}\label{e:FIN}
Given $X\neq_E \emptyset$: $$Fin(X)\longleftrightarrow$$ $$\exists
n(n\in\omega\wedge\forall\gamma(CD_X(\gamma)\longrightarrow\exists
F(F\subseteq \gamma\times n^{+}\wedge qf(F)\wedge  $$ $$<n;X>\in
F\wedge\forall z(z\in\gamma\longrightarrow\exists j(j\in
n^{+}\wedge< j;z>\in F))
 \wedge$$ $$\forall j(j\in
n^{+}\wedge j\neq_E 0\longrightarrow DD_{F(j)}(F(j-1)))))$$

$Fin(X)$ means: X is finite.
\end{definition}

The intuitive interpretation of this definition is that a quasiset
will be considered to be finite if and only if all its descendant
chains last (reach the empty set) in a finite series of steps. In
the following, we will make use of another axiom:

\begin{axiom}\label{a:H2}
$$\forall_Q X\forall_Q Y(Fin(X)\wedge Fin(Y)\wedge Y\equiv
X\longrightarrow$$ $$((X\subseteq Y\longrightarrow Y=_E X)\wedge
(Y\subseteq X\longrightarrow Y=_E X)))$$
\end{axiom}
This postulate means that if a given quasiset is finite, any of
its direct descendants is distinguishable from the original
quasiset. This is a reasonable assumption, because the only thing
that distinguishes two collections of truly indistinguishable
objects of the same class should be their quantity of elements. We
will define the quasicardinal to fit this intuitive idea. We will
need the following proposition:

\begin{prop}\label{e:FDE}
Let $X$ be a quasiset satisfying $X\neq_E\emptyset$, $Fin(X)$ y
$CD_X(\gamma)$. Let $F$ be a quasifuntion as in \ref{e:FIN}. Then
$F(0)=_E\emptyset$
\end{prop}

\begin{proof}
As $F$ is a q-function it follows that there exists
$\lambda\in\gamma$ such that $<0;\lambda>\in F$. If $\lambda$ were
not the empty quasiset then there would exist a direct descendant
of $\lambda$, call it $\lambda^{-}$, and this last would pertain
to $\gamma$, for $DD_X(\gamma)$. By hypothesis, there exists $j\in
n^{+}$ which satisfies $<j;\lambda^{-}>\in F$. If $j=0$ then
$<0;\lambda^{-}>\in F\wedge<0;\lambda>\in F$. But
$\lambda^{-}\subset\lambda$  and so it follows that
$\neg(\lambda^{-}\equiv\lambda)$ (using axiom \ref{a:H2}) and so
$F$ would not be a quasifunction any more. Thereafter $j\neq 0$,
and in that case, $j>0$. So $<j;\lambda^{-}>\in F$ and
$<0;\lambda>\in F$ by construction of $F$, as $j>0$ it follows
that $\lambda\subset\lambda^{-}$. But this is a contradiction. The
contradiction comes from the supposition that
$\lambda\neq_E\emptyset$ and then, we have $<0;\emptyset>\in F$.
(This is the same that $F(0)=\emptyset$.) \qed
\\
\end{proof}

\begin{prop}
Let X be a non empty and finite quasiset. Then there exists a
unique $n$ which satisfies definition \ref{e:FIN}.
\end{prop}

\begin{proof} Let us suppose that there are two integers $n$ and $m$, $n<m$,
satisfying definition \ref{e:FIN}. In this case, given a
descendant chain $\gamma$, there exist $F$ and $F'$  which
satisfy:

$$(a)F(n)=_E X ,(b) F\subseteq\gamma\times n^{+},(c) \forall
j(DD_{F(j)}(F(j-1)))$$ $$(a')F'(m)=_E X, (b')
F'\subseteq\gamma\times m^{+},(c') \forall j(
DD_{F'(j)}(F'(j-1)))$$ Moreover, by proposition \ref{e:FDE} we
have that $F(0)=_E \emptyset$ and $F'(0)=_E \emptyset$ and by
hypothesis $m-n>0$. On the other hand, $$F(n)=_E X=_E
F'(m)\longrightarrow F(n-1)=_E X^{-}=_E F'(m-1)$$ This last is
true because there is a single direct descendant of $X$ in
$\gamma$ and because of the characteristics of quasifunctions. Due
to the fact that this is a process with a finite number of steps,
we can continue extracting until $0$ is reached, and for that
reason, we can arrange the $F(j)$s in a table:

$$F(n)=_E X=_E F'(m)$$ $$F(n-1)=_E X^{-}=_E F'(m-1)$$ $$\vdots$$
$$F(0)=_E F(n-n)=_E\emptyset=_E F'(m-n)$$ But $j=m-n\neq0$ implies
$F'(j)=_E\emptyset$ with $j\neq 0$ and $F'(0)=_E\emptyset$ (by
proposition \ref{e:FDE}). This is absurd because it would imply
$\emptyset\subset\emptyset$, and we know that $\emptyset$ cannot
be a proper subquasiset of $\emptyset$. The contradiction comes
from the supposition that $m-n>0$ and so $m=n$. \qed
\\
\end{proof}

The last proposition justifies the following definition:

\begin{definition}\label{d:DEFFIN}
If $X\neq_E \emptyset$ and $X$ is finite, we say that its
quasicardinal is the only natural number $n$ according to
definition \ref{e:FIN}. Then we write $qcard(X)=_En$. If $X$ is
the empty quasiset, we say that it is finite too and its
quasicardinal is zero.
\end{definition}

It is important at this point to make the following remark. In the
last definition we used the symbol $qcard()$ to express that the
quasicardinal of $X$ is $n$. Note the distinction of this derived
unary function and the primitive unary function $qc()$ used in
\cite{Un estudio}. In the next subsection, we show that $qcard()$
has the same properties as $qc()$ when restricted to finite
quasisets.

\subsection{Quasicardinality Theorems}

Now, we are ready to prove as theorems all the axioms of
quasicardinality of $Q$ for the case of finite quasisets. In the
rest of this section, we will assume that all quasisets are
finite, unless the contrary is mentioned.

\begin{theo}
$$\forall_Q X(X\neq_E \emptyset\longrightarrow qcard(X)\neq_E
\emptyset)$$
\end{theo}
\begin{proof}

As $X\neq_E\emptyset$, it follows that there exists at least one
direct descendant of $X$, and each descendant chain has at least
two elements (they surely contain $X$ and $\emptyset$). So the
quasifunctions $F$ are defined starting from $0^{+}$ and then,
$qcard(X)\neq_E 0$. \qed
\\
\end{proof}

\begin{theo}
$$\forall_Q
X(qcard(X)=_E\alpha\longrightarrow\forall\beta(\beta\leq\alpha\longrightarrow\exists_Q
Y(Y\subseteq X\wedge qcard(Y)=_E\beta)))$$
\end{theo}

\begin{proof}
Suppose that $qcard(X)=_E n$ y $m<n$. Let $\gamma$ be a descendant
chain of X. So, there exists $F$ such that $F\subseteq\gamma\times
n^{+}$, $F(n)=_E X$ y $DD_{F(j)}(F(j-1))$. Furthermore,
$F(0)=_E\emptyset$.\\ Consider $Y=_E F(m)$. Then $Y\in\gamma$. It
is easy to convince oneself that the union of each one of the
descendant chains of Y with the sets that precede Y in $\gamma$
are descendant chains of $X$. Then, given an arbitrary descendant
chain in Y, let us extend it to a descendant chain in X and
restrict the quasifunction corresponding to the last one
(according to definition \ref{e:FIN}) to the quasiset  which
contains $Y$ and its descendants. Then, it follows that $Y$ has
quasicardinal $m$. \qed
\\
\end{proof}

\begin{theo}
$$\forall_Q X\forall_Q Y(Fin(X)\wedge X\subset Y\longrightarrow
qcard(X)<qcard(Y))$$
\end{theo}

\begin{proof}
Suppose that $qcard(X)\geq n=_E qcard(Y)$. It is clear that:
$$Y=_E(Y\setminus X)\cup X$$ If $X\neq_E\emptyset$ it follows that
there exists $z$ such that $z\in X$.
Then:$$Y\setminus<z>=_E(Y\setminus X)\cup(X\setminus<z>)$$ or with
the usual notation:$$Y^{-} =_E (Y\setminus X) \cup (X^{-})$$
Taking this into account, it is apparent that the descendant chain
in $X$ induces a descendant chain in $Y$:

$$X\supset X^{-}\supset
X^{--}\supset\cdots\cdots\supset\emptyset\longrightarrow$$
$$\longrightarrow Y\supset Y^{-}\supset
Y^{--}\supset\cdots\cdots\cdots\cdots$$\\ with $Y=_E(Y\setminus
X)\cup X$, $Y^{-}=_E(Y\setminus X)\cup (X^{-})$,
$Y^{--}=_E(Y\setminus X)\cup (X^{--})$, etc.\\ If the empty set is
reached in $m=_E n$ steps, then we obtain:
$$\emptyset=_E(Y\setminus X)\cup\emptyset\longrightarrow
Y\setminus X=_E\emptyset\longrightarrow Y=_E X$$ but this has no
sense. If the empty quasiset is reached in $m>n$ steps, it is
equally absurd, because we would arrive at zero in a number of
steps greater than that needed for any chain of $Y$. So we see
that $X$ must be finite too, for if it were not so, it would have
a chain in $Y$ with an associated number greater than $n$ . \qed
\\
\end{proof}

\begin{theo}\label{e:SUM}
$$\forall_Q X\forall_Q Y(\forall w(w\notin X\vee w\notin
Y)\longrightarrow$$ $$(qcard(X\cup Y)=_E qcard(X)+ qcard(Y))$$
\end{theo}
\begin{proof}
First, let us see that the union of finite quasisets gives a
finite quasiset. Let $Z=_E X\cup Y$.
$Z\neq_E\emptyset\longrightarrow\exists w(w\in
Z)\longrightarrow(w\in X\vee w\in Y)$ but not both. If $w\in X$,
we have $$Z^{-}=_E Z\setminus<w>=_E(X\setminus<w>)\cup Y=_E
(X^{-})\cup Y$$ (if $w$ belongs to $Y$, we obtain an analogous
equation). Taking a direct descendent of $Z^{-}$, we obtain one of
two possibilities:
 $$Z^{--}=_E(X^{--})\cup Y$$
or:
 $$ Z^{--}=_E(X^{-})\cup(Y^{-})$$
$X$ and $Y$ are finite quasisets, so let us suppose that $X$ has
quasicardinal $m$ and that $Y$ has quasicardinal $n$. So all the
chains of $X$ vanish in $m$ steps and all the chains of $Y$ in $n$
steps. Then, we see that each element in a chain of $Z$, call it
$e$, can be expressed as $$e=_E\alpha\cup\beta$$ where $\alpha$ is
an element of some chain in $X$ and $\beta$ in $Y$. Then, any
chain in $Z$ vanishes in $m+n$ steps. This fact proves that the
quasicardinal of $Z$ is $m+n$. In all these demonstrations, the
expressions $n$, $m$ and $n+m$ steps, have to be associated with
the quasifunctions of definition \ref{e:FIN}. We point out once
more that we are always dealing with finite processes. \qed
\\
\end{proof}

In the following, we will use the following proposition:

\begin{prop}\label{e:ZZZ}
Let X be a quasiset such that $qcard(X)=_E1$. Then it follows that
there exists $z$ satisfying $z\in X$ and $X=_E<z>$.
\end{prop}

\begin{proof}
Let $\gamma$ be a chain of $X$. Then there exists a quasifunction
$F$ satisfying $F(1)=_E X$ and (so) $F(0)=_E\emptyset$.
Furthermore, $F$ has the following property: $DD_{F(j)}(F(j-1))$,
for any $j\neq_E 0$. Then we have $DD_{F(1)}(F(0))$ and for that
reason $DD_{X}(\emptyset)$. From this, the existence of $z$ such
that $z\in X$ and $X\setminus<z>=_E\emptyset$  follows. But in
this case, there is no other choice than $X=_E<z>$. \qed \\
\end{proof}

\begin{theo}
$$\forall_Q X (qcard(\wp (X))=_E 2^{qcard(X)})$$
\end{theo}

\begin{proof}
By induction on the quasicardinal. Let us suppose that $X$ has
quasicardinal $1$. Then, by proposition \ref{e:ZZZ} it follows
that $\exists z$ with $X=_E <z>$. Thus, by proposition
\ref{e:LABASE}, the only possible subquasisets are $\emptyset$ and
$X$, i.e., $<z>$. It is easy to convince oneself that all the
chains of $\wp (X)$ vanish in $2$ steps, and therefore it follows
that $qcard(\wp (X))=_E 2=_E 2^{qcard(X)}$.\\ Suppose now than the
assumption is true for any quasiset of quasicardinal $n$, and let
us see that this is true for $n+1$. Let $X$ be an arbitrary
quasiset of quasicardinal $n+1$ and $\gamma$ a chain of this
quasiset. Then, it follows that $qcard(X^{-})=_E n$, for if
$X^{-}$ is direct descendant of $X$, we have that there exists $z$
with $z\in X$ and $X=_E X^{-}\cup<z>$. For this reason:
$$qcard(X)=_E qcard(X^{-})+qcard(<z>)\longrightarrow n+1=_E
qcard(X^{-})+1$$ But in this case we have $qcard(P(X^{-}))=_E
2^{n}$ and as $P(X^{-})\subseteq \wp (X)$ it follows that there
exist at least $2^{n}$ different quasisets $\wp (X)$, counting the
empty quasiset too. But now we see that there remain $2^{n}$
quasisets formed by the union of each one of the quasisets of
$P(X^{-})$ with $<z>$. In this case, we see that $\wp (X)$ has at
least $2^{n}+2^{n}=_E 2^{n+1}$ elements. To see that we have
counted all the possible quasisets (all the elements of $\wp (X)$)
we proceed as follows. Let $Y$ be a quasiset such that $Y\in \wp
(X)$. Then, only one of two things may happen:
$Y\cap<z>=_E\emptyset$ or $\neg(Y\cap<z>=_E\emptyset)$. If the
first option is true, it follows that $Y\subseteq X^{-}$ and so it
has been already counted. If the second one is true, we have
$Y\cap<z>=_E<z>$ (using proposition \ref{e:LABASE}), and then
$Y=_E Y\setminus<z>\cup<z>$ and this is to say that $Y$ is the
union of an element of $P(X^{-})$ with $<z>$ and then, it has been
already counted.\qed
\\
\end{proof}

\subsection{On the existence of finite quasisets}\label{s:PECF}

We have seen that if there exist finite quasisets in the sense
given in this section, these would satisfy the axioms of
quasicardinality of $Q$. Now, it is worthwhile to ask: under which
conditions do pure and finite quasisets exist?

Assume that there exists a non empty and pure quasiset. Call it
$X$. Then, there exists $z$ such that $z\in X$, and its singleton
$<z>$ exists too. Which is the quasicardinal of $<z>$? The
following theorem gives the answer.

\begin{theo}
Let $X$ be a non empty quasiset and let $<z>$ be a singleton of
$X$. Then, $\exists qcard(<z>)$ and $qcard(<z>)=_{E}1$.
\end{theo}

\begin{proof}
Let $\gamma$ be a descendant chain of $<z>$. By proposition
\ref{e:LABASE} we have that every direct descendant of $<z>$ is
the empty quasiset. Then the members of $\gamma$ can be shown
explicitly: $\gamma=_E [y\in \wp (X) : y=_E <z> \vee y=_E
\emptyset]$. So let $F$, be the following quasifunction (expressed
as a collection of ordered pairs): $$F=_E
[<1;<z>>;<0;\emptyset>]$$ $F$ satisfies all the conditions listed
in definition \ref{e:FIN} for $n=_E 1$. As $\gamma$ is an
arbitrary chain, it follows that, by definition of quasicardinal,
$qcard(<z>) =_E1$ \qed
\end{proof}

Thus we see that the existence of a pure quasiset with
quasicardinal equal to $1$ follows directly from the assumption of
the existence of some pure non empty quasiset. What do we need to
assume for granted the existence of pure quasisets with an
arbitrary but finite quasicardinal? To answer this question we
first give the following definition:

\begin{definition}
A pure quasiset $X$ is infinite if it is not finite.
\end{definition}

If we assume that there exists an infinite quasiset then, we have
a ``source'' of infinite and disjoint singletons, because to be
infinite implies that there is at least one descendent chain for
which it is not possible to find a quasifunction $F$ as defined in
\ref{e:FIN}. But with this arbitrarily big collection of
singletons, we can make use of theorem \ref{e:SUM} as follows. If
$<z>$ is a singleton of the chain and $<w>$ is another one which
satisfies $<z>\cap<w>=_E\emptyset$, we then have, by theorem
\ref{e:SUM}, that: $$qcard(<z>\cup<w>)=_E qcard(<z>)+qcard(<w>)=_E
1+1=_E 2$$ It is obvious that, with this procedure, we can
generate finite and pure quasisets with any desired quasicardinal.

\subsection{What do we mean by ``counting''?}

What do we mean when we say that a given set in $ZF$ has ten
elements? In that theory we say that its cardinal is ten, i.e., it
is equivalent to the ordinal number $10$. If the set is infinite,
we say that it has a cardinal, because we can prove that it can be
well ordered, and so it is equivalent to a set of ordinals which
has a minimal element (which in turn, is the cardinal of the given
set) \cite{Halmos}. This fact lets us assert that every set of
$ZF$ has an associated cardinal, and the cardinal is interpreted
as the number of elements of the set.

So we see that in $ZF$, the formalized idea of quantity of
elements bases itself in the concept of equivalence (the existence
of a bijective function between sets). We see that this aspect of
the formal construction of cardinality in $ZF$ is in
correspondence with the intuitive idea of quantity that we deal in
our every day experience: two sets have the same quantity of
elements if we can put them in biunivocal correspondence. But we
can detach of intuitive interpretations, and say that possessing a
cardinal (in the sense of $ZF$) is simply a property of the sets
deduced from the axioms of $ZF$. Then it is found that the formal
development of the theory shows properties which depart radically
from our intuition, as the fact that there are many classes of
infinity, a non trivial result of set theory \cite{Halmos}. Thus
we see how formalism development enriches our concept of infinite,
showing properties which are by other means inaccessible to
intuitive reasoning.

What happens in $Q$? As it can be proved that pure quasisets
cannot be well ordered \cite{Un estudio}, it follows that the
argument used in $ZF$ to prove that to every set corresponds a
single cardinal no longer holds. This means that in $Q$ we cannot
use the same technique as in $ZF$ if we want to assign a
quasicardinal to every quasiset in a consistent manner.

Is there any other technique which permits to assign a
quasicardinal to any arbitrary given quasiset taking quasicardinal
as a derived concept? We have in part given an answer to this
query because, as shown in section \ref{s:Reformulation}, we can
assign a derived quasicardinal to a certain class of quasisets
(finite quasisets). This assignment rule is reasonable in the
following sense: if it is applied to finite sets (elements of
$ZFU$), the quasicardinal of definition \ref{d:DEFFIN} coincides
with the one of $ZF$, as can be readily checked out. Furthermore,
if we restrict ourselves to the classical part of $Q$, axioms
\ref{a:H1} and \ref{a:H2} are no longer needed, for they become
theorems of $ZF$ axioms.

Postulates \ref{a:H1} and \ref{a:H2} exposed in section
\ref{s:Reformulation} are interpreted as the translation, to the
language of $Q$, of the idea of extracting the elements one by
one. But this method is valid only for finite quasisets. The trick
of extracting the elements one by one until the empty quasiset is
reached fails for every thing which deserves to be called
infinite. But we make the following observation: it is possible to
find in $ZF$ examples of sets which possess non denumerable
descendant chains. So, it happens that the concept of descendant
chain is not only applicable to denumerable process (as a first
glance may suggest), and this may be useful to extend the
definition of quasicardinal to infinite quasisets. Thought we are
not prepared to solve the problem now, by taking into account the
discussion presented above we can formulate the following problem
in $Q$: it would be interesting to formulate an idea of infinite
and pure quasisets (without using the axioms of quasicardinality
of $Q$), and then find a way to compare them. This work would be
interesting because the intent to solve it will probably give rise
to new techniques (distinct and perhaps inequivalent to the ones
known in $ZF$) of thinking the infinite, and so, enrich our
knowledge of this intricate concept. We think that this is an
interesting line of work because obtaining satisfactory results in
this direction would be an advance in the development of a
generalized theory of cardinality. Furthermore, the Manin's
problem can be extended to the quest of developing a ``set
theory'' in which some ``sets'' do not possess an associated
cardinal. For the idea that every set must have an associated
number of elements is based in our every day experience as well as
the idea that its elements have to be individuals.

In this section we have proved that in a theory concerning
collections of truly indistinguishable objects (as $Q$), the
quantity of elements, has not necessarily be taken as a primitive
concept. This result encourages the search for a theory in which
it is impossible to assign a quasicardinal to certain quasisets in
a consistent manner, thus allowing to describe what it seems to
happen with some quantum systems, in which non individuality
expresses itself in the fact that particle number is not defined,
besides ontological indistinguishability. We will explore
elsewhere the possibility of using the non classical part of $Q$
without modifying the axioms listed in \cite{Un estudio}, to make
a construction analogous to the state space used in quantum
mechanics in order to incorporate undefined particle number
systems in the formalism.

\section{Conclusions}\label{s:Conclusions}
In this work, we have searched for a $Q$-like theory in which the
quasicardinal is not taken as a primitive concept alike $Q$. In
such a theory, to have a well defined quasicardinal would
eventually be a property that some quasisets possess, and others
do not, thus allowing the existence of quasisets for which their
quantity of elements is not well defined. This search was
motivated in the discussions given in sections \ref{s:QaR} and
\ref{s:Individuality}. For in \ref{s:QaR} we discussed the
existence of quantum systems with not well defined particle number
and studied whether they can be represented as quasisets. We saw
that this cannot be done, because every quasiset has a well
defined quasicardinal, and suggested that a possible way out is to
take quasicardinal as a derived concept. In section
\ref{s:Individuality}, we found other motivations for considering
quasicardinal as a derived concept. This was done regarding
particle number as a result of the measurement process, in the
sense that this process may be considered as the basis for the
particle number interpretation. In \ref{s:Reformulation} we saw
that it is possible to develop a theory about collections of
indistinguishable entities in which quasicardinal is not a
primitive concept. In doing this, we transcribed to the first
order language of $Q$ the idea of counting the elements of a
collection extracting their elements one by one, without knowing
which is which (because in $Q$ this query is not defined). Thus,
we introduced the concept of descendant chains. At the end of the
section we discussed the possibility of developing this method to
extend the definition to infinite sets.

In the construction shown in \ref{s:Reformulation} we have
reobtained that every (finite) quasiset has a well defined
quasicardinal. Perhaps a deeper transformation in the axiomatic of
$Q$ is required to obtain the desired result of undefined
quasicardinal. One way would be to explore the possibility of
developing not completely defined belonging relations. This is
done in Quaset theory (see \cite{Dalla Chiara}), but quasicardinal
is assumed as a primitive concept there, as in $Q$. Another
interesting proposal is to modify logical quantifiers such as
``$\exists$'' and ``$\forall$'' \cite{coskra}. The modification
could be made in such a way that the property of undefined
quasicardinal appear as a product of it.

The axiomatic variant exposed in \ref{s:Reformulation} shows
explicitly that in a theory about collections of indistinguishable
entities, the quasicardinal needs not be necessarily taken as a
primitive concept. This result encourages the research of more
complex axiomatic formulations, able to incorporate the quantum
systems with undefined particle number as sets of some kind, thus
enriching the Manin's problem.

\vskip1truecm

\noindent {\bf Aknowledgements} \noindent The authors thank Hector
Freytes for useful discussions. They also thank the careful
reading of the anonymous referees.  F. H. also thanks Newton da
Costa and D\'ecio Krause for their interesting commentaries.
\noindent This work was partially supported by the following
grants: PICT 04-17687 (ANPCyT), PIP N$^o$ 6461/05 (CONICET),
UBACyT N$^o$ X081.

\end{document}